\documentclass{article}
\usepackage{amsmath}
\usepackage{amssymb}
\usepackage{latexsym}

\usepackage[all]{xy}
\makeatletter
\renewcommand\section{\@startsection {section}{1}{\z@}%
                                   {-3.5ex \@plus -1ex \@minus -.2ex}%
                                   {2.3ex \@plus.2ex}%
                                   {\normalfont\large\bfseries}}
\makeatother
\usepackage{theorem}
\theorembodyfont{\normalfont\itshape}
\theoremstyle{change}
\newtheorem{lemma}{Lemma.}[section]
\newtheorem{prop}[lemma]{Proposition.}

\theorembodyfont{\normalfont}

\newtheorem{taller}[lemma]{$\!\!$}

\newenvironment{blanko}[1]%
{\begin{taller}{\normalfont\bfseries #1}\normalfont}%
{\end{taller}}

\providecommand{\qed}{\hspace*{\fill}$\Box$}

\newenvironment{proof}{%
\begin{list}{\em Proof. }%
{\setlength{\labelsep}{0mm}\setlength{\leftmargin}{0mm}%
\setlength{\labelwidth}{0mm}\setlength{\listparindent}{\parindent}%
\setlength{\parsep}{\parskip}\setlength{\partopsep}{0mm}}%
\item}{\qed\end{list}}

\newenvironment{proof*}[1]{%
\begin{list}{\em #1 }%
{\setlength{\labelsep}{0mm}\setlength{\leftmargin}{0mm}%
\setlength{\labelwidth}{0mm}\setlength{\listparindent}{\parindent}%
\setlength{\parsep}{\parskip}\setlength{\partopsep}{0mm}}%
\item}{\qed\end{list}}

\usepackage{mathrsfs}

\newcommand{\CC}{\mathscr{C}}
\newcommand{\C}{\mathbb{C}}
\newcommand{\tensor}{\otimes}
\newcommand{\ground}{\Bbbk}
\newcommand{\id}{\operatorname{id}}
\newcommand{\ind}{\operatorname{in}}
\newcommand{\out}{\operatorname{out}}
\newcommand{\res}{\operatorname{res}}
\newcommand{\Ker}{\operatorname{Ker}}
\renewcommand{\Im}{\operatorname{Im}}
\newcommand{\Lin}{\operatorname{Lin}}

\renewcommand{\epsilon}{\varepsilon}

\usepackage{texdraw}

\everytexdraw{%
	\drawdim pt \linewd 0.5 \textref h:C v:C
	\setunitscale 1
}

\newcommand{\smalldot}{
  \bsegment
    \move (0 0) \fcir f:0 r:1.5
  \esegment
}

\newcommand{\trevert}{
\bsegment
\move (3 0) \lvec (6 0) \lvec (9 3) \move (6 0) \lvec (9 -3)
\esegment
}

\newcommand{\trekant}{
\bsegment \setunitscale{0.8}
\move (2 0) \lvec (6 0) \lvec (13 5)
\move (6 0) \lvec (13 -5)
\move (10 2.7) \lvec (10 -2.7) 
\esegment
}
\newcommand{\femkant}{
\bsegment \setunitscale{0.8}
\move (2 0) \lvec (6 0) \lvec (20 10)
\move (6 0) \lvec (20 -10)
\move (10 2.7) \lvec (10 -2.7) 
\move (16 7) \lvec (16 -7) 
\esegment
}

\newcommand{\tinydot}{
  \bsegment
    \move (0 0) \fcir f:0 r:1.1
  \esegment
}

\newcommand{\twotree}{
\bsegment 
\move (0 -4) \lvec (0 0) \tinydot
\lvec (-2 5) \tinydot
\lvec (-4 10)
\move (-2 5) \lvec (0 10)
\move (0 0) \lvec (3 10)
\esegment
}

\newcommand{\onetree}{
\bsegment 
\move (0 -4) \lvec (0 0) \tinydot
\lvec (-2 5) 
\move (0 0) \lvec (2 5)
\esegment
}

\newcommand{\zerotree}{
\bsegment 
\move (0 -4) \lvec (0 4)
\esegment
}

\newcommand{\twonodetree}{
\bsegment 
\move (0 -4) \smalldot \lvec (0 4) \smalldot
\esegment
}

\newcommand{\vtree}{
\bsegment 
\move (0 -4) \smalldot \lvec (-3 4) \smalldot
\move (0 -4) \lvec (3 4) \smalldot
\esegment
}

\begin{document}

\title{Perturbative renormalisation for \\[2pt] not-quite-connected bialgebras}

\author{Joachim Kock}

\date{}

\maketitle

\begin{abstract}
  We observe that the Connes--Kreimer Hopf-algebraic approach to perturbative
  renormalisation works not just for Hopf algebras but more generally for
  filtered bialgebras $B$ with the property that $B_0$ is spanned by group-like
  elements (e.g.~pointed bialgebras with the coradical filtration).
  Such bialgebras occur naturally both in Quantum Field Theory,
  where they have some attractive features, and
  elsewhere in Combinatorics, where they cover a comprehensive class of 
  incidence bialgebras.  In particular, the setting allows us to interpret
  M\"obius inversion as an instance of renormalisation.
\end{abstract}

\section{Introduction}

Kreimer~\cite{Kreimer:9707029} made the crucial discovery that the combinatorics
underlying the 
BPHZ
renormalisation scheme in perturbative quantum field theory can be encoded in a
Hopf algebra, and his seminal joint work with
Connes~\cite{Connes-Kreimer:9808042},~\cite{Connes-Kreimer:9912092} highlighted
the significance of this through deep connections to many areas of mathematics,
constituting a starting point for numerous further developments.

Subsequent work by Ebrahimi-Fard, Guo,
Manchon~\cite{EbrahimiFard-Guo-Manchon:0602004} and others provided a more
algebraic formulation of the Connes--Kreimer approach, expressing it abstractly
in the setting of a connected graded Hopf algebra $H$ and a Rota--Baxter algebra
$A$, as briefly recalled in Section~\ref{sec:Hopf} below.

The present note makes the observation that the same construction works when the
Hopf algebra is replaced by a filtered bialgebra $B$ with the property that $B_0$
is spanned by group-like elements.  Interest in this observation resides in
the fact that there are natural examples in perturbative QFT where the
bialgebra $B$ contains interesting combinatorics of physical relevance,
invisible in the quotient Hopf algebra $H$.  In Combinatorics, the notion covers
incidence bialgebras, and we show that in this setting M\"obius inversion
becomes a special instance of renormalisation.

\bigskip

In Section~\ref{sec:Hopf} we quickly run through the Hopf case, to set up
notation, and to facilitate the generalisation to bialgebras.  This
generalisation comes in two versions.  In Section~\ref{sec:bialgI} we take the
simplest approach, requiring only that $B_0$ is spanned by group-like elements
$x$, but assuming that the Feynman rules $\phi:B\to A$ satisfy $\phi(x)=1$.
However, this assumption on the Feynman rules is strictly speaking not realistic
physically.  In Section~\ref{sec:bialgII} we give a second version, where $B$ is
assumed to be a polynomial algebra generated by elements whose comultiplication
has group-like components in degree $0$.  (See \ref{condII} for the precise
condition.)  This condition, which is essentially satisfied automatically for
bialgebras defined in terms of combinatorial data, ensures the existence of a
residue operator, which allows to calibrate the Feynman rules, now allowed to
take invertible values on group-like elements, so as to reduce to the previous
case, $\tilde\phi(x)=1$.  In Section~\ref{sec:ex}, the relevance of the
bialgebra generalisation is substantiated through interpretation in perturbative
QFT, and with some related examples from Combinatorics.  Finally in
Section~\ref{sec:M}, we establish the renormalisation principle for coalgebras,
show that incidence coalgebras of M\"obius categories (and more generally, of
M\"obius decomposition spaces) constitute examples, and show that M\"obius
inversion is a special case of renormalisation.

\section{Hopf algebra renormalisation}
\label{sec:Hopf}

\begin{blanko}{Connes--Kreimer Hopf-algebraic renormalisation \cite{Connes-Kreimer:9912092}.}
  (Convenient self-contained accounts are given in \cite{Manchon:0408}, 
  generous with mathematical preliminaries, and in
  \cite{EbrahimiFard-Kreimer:0510}, emphasising physical background and 
  perspectives.)  Let $H$ be a connected graded Hopf algebra, let
  $(A,R)$ be a 
  Rota--Baxter algebra, with idempotent Rota--Baxter
  operator $R$ of weight $1$.  Put $A_+ := \Ker(R)$ (a unital subalgebra of $A)$
  and $A_- :=\Im(R)$ (a non-unital subalgebra of $A$); we have $A = A_-
  \oplus A_+$.

  The space $\Lin(H,A)$ of linear maps is a monoid under the convolution
  product, with $e := \eta_A \circ \epsilon$ as neutral element.  Linear maps
  $\phi:H\to A$ with $\phi(1)=1$ are referred to as (regularised) Feynman rules;
  these form a group.  The Feynman rules that are furthermore algebra
  homomorphisms form a subgroup, the group of {\em $A$-valued characters} of $H$.
\end{blanko}

\begin{prop}[\cite{Connes-Kreimer:9912092}]\label{RENH}
  For each Feynman rule $\phi$, denote by $\phi_-$ the linear map defined recursively by
  \begin{equation}\label{phiH}
    \phi_- \ := \ e + R\big(\phi_- * (e\!-\!\phi)\big) .
  \end{equation}
  The {\em renormalised Feynman rule}
  $$
  \phi_+ := \phi_- * \phi
  $$
  maps $H_+$ into $A_+$.
  Furthermore, if $\phi$ is a character, then so are $\phi_-$ and $\phi_+$.
\end{prop}
(The equation
$
\phi = \phi_-^{-1} * \phi_+
$
constitutes a Birkhoff decomposition of $\phi$, as has been observed and exploited by 
Connes and Kreimer; see also \cite{Connes-Marcolli}.)

For a proof of the proposition, see Manchon~\cite{Manchon:0408}.    In the proof, the
antipode of $H$ does not play any role (although other formulations may exploit it,
cf.~Kreimer's interpretation of $\phi_-$ as a twisted antipode).  More important
is the fact that in a connected graded bialgebra, the comultiplication takes the
following form (for $x\in H_+$):
\begin{equation}\label{eq:primitive}
  \Delta(x) = 1 \tensor x + \sum_{(x)} x'\tensor x'' + x \tensor 1.
\end{equation}
The middle sum is restricted Sweedler notation for all the terms of the
comultiplication of positive degree strictly less than the degree of $x$.
This splitting makes the recursive definition of $\phi_-$ meaningful: $e$ and
$ \phi$ agree on $H_0$, so this case is the basis of the recursive definition,
$\phi_-(1) = e(1)=1$.  For $x$ of degree $n>0$, we have $e(x)=0$, and hence
\begin{equation}\label{phi-x}
\phi_-(x) = - R \phi(x) -R\big( \sum_{(x)} 
  \phi_-(x')\phi(x'')     \big) - 0 ,
\end{equation}
where the middle sum involves only elements of degree strictly less than $n$,
and hence are determined inductively.
  
The Rota--Baxter axiom is needed only for the character property:
to show that $\phi_-(xy) = \phi_-(x)\phi_-(y)$ one can clearly assume
that both $x$ and $y$ are of positive degree, and then exploit 
\eqref{eq:primitive} in an inductive argument where the Rota--Baxter
property turns out be exactly what is needed.

\begin{blanko}{Example: QFT.}\label{ex:QFT}
  In perturbative quantum field theory \cite{Connes-Kreimer:9912092}
  (see Section~\ref{sec:ex} for further
  details), $H$ is a Hopf algebra of certain 1PI Feynman graphs (excluding
  graphs with no inner lines), $\phi: H\to
  A:=\C[[t,t^{-1}]$ is a dimension-regularised Feynman rule,
  and $R$ is taking pole part, the minimal
  subtraction scheme.  Other regularisations and renormalisation schemes fit the
  description too.
\end{blanko}

\section{Bialgebra renormalisation I}

\label{sec:bialgI}

We now weaken the hypotheses.  

\begin{blanko}{Hypothesis I.}
  Instead of a connected graded Hopf algebra, we work with a filtered bialgebra
  $B$ with the property that $B_0$ is spanned by group-like elements.
\end{blanko}

\begin{blanko}{Remark on pointedness.}\footnote{See
  Sweedler~\cite{Sweedler} for the notions of pointedness and the coradical 
  filtration (not needed in what follows).}
  One can show that any bialgebra $B$ satisfying Hypothesis I is in fact pointed.
  Conversely, for every pointed bialgebra $B$, for the coradical filtration we
  have that $B_0$ is spanned by group-like elements.  Hence an alternative to
  Hypothesis I is to work with pointed bialgebras.
  The condition in Hypothesis I is preferred over pointedness
  because in some cases of interest,
  the filtration may be different from the coradical
  filtration.
\end{blanko}

In many important cases, $B$ will actually be graded.  (The more general
hypothesis will be important to cover many examples from Combinatorics,
cf.~\ref{inc}--\ref{3=1+1} below.)  Even if $B$ is only filtered, some of the
main arguments exploit degree:

\begin{blanko}{Auxiliary notion of degree.}
  Denote the filtration with subscripts:
  $$B_0 \subset B_1 \subset B_2 \subset \cdots \subset B.$$
  Put $B(0):= B_0$, and for each $i\geq 0$, 
  choose a linear complement $B(i+1)$ of $B_i \subset B_{i+1}$
  so as to write $B_{i+1} = B_i 
  \oplus B(i+1)$.  Altogether $B_n = \oplus_{i=0}^n B(i)$  and
  $$
  B = \oplus_{n=0}^\infty B(n),
  $$
  defining a notion of degree.
  We shall assume that each $B(i+1)$ is chosen inside $\Ker \epsilon$;
  this is possible since $1\oplus \Ker \epsilon = B$  and  $1\in B_0$.
  We put $B_+ := \oplus_{n=1}^\infty B(n) \subset \Ker\epsilon$.

  With this auxiliary notion of degree, we can write the comultiplication of a
  homogeneous degree-$n$ element $x$ according to degree splitting.
  Denote by $\Delta_{p,q}(x)$ the projection of $\Delta(x)$ onto $B(p)\tensor 
  B(q)$, then we can write
  $$
  \Delta(x) = \sum_{p+q\leq n} \Delta_{p,q}(x) .
  $$
  Note that this may involve terms of lower degree than expected, hence the
  summation over $p+q\leq n$, in contrast to the graded case where 
  only terms with $p+q=n$ contribute.
\end{blanko}

\begin{lemma}
  Assuming Hypothesis I, and with respect to the auxiliary degree,
  for $\deg(x) = n>0$ we have
  \begin{equation}\label{eq:split}
  \Delta(x) \ = \ \Delta_{0,n}(x) \ + \ \sum_{(x)} x'\tensor x'' \ + \ 
  \Delta_{n,0}(x) .
  \end{equation}
  The middle part (indicated 
  with restricted Sweedler notation as in \eqref{eq:primitive}) involves
  only $x'$ and $x''$ of positive degree
  strictly less than $n$.
\end{lemma}
  The non-trivial  statement is that no degree splittings of
  type $0+m$ or $m+0$ occur with $m<n$.  This is a consequence of counitality,
  together with the assumption that $B_+ \subset \Ker \epsilon$.
  
\begin{lemma}
  The ideal $I =\langle \;1-x \mid x \text{ \rm group-like}\;\rangle \subset B$
  is also a (filtered) two-sided co-ideal, and the quotient bialgebra $H:=B/I$ is
  connected, hence Hopf.
\end{lemma}

\begin{lemma}
  For $A$ a unital algebra, the linear maps $\phi\in\Lin(B,A)$ such that
  $\phi(x)=1$ for all group-like elements $x$, form a group under convolution.
  The subgroup of multiplicative maps is
  isomorphic to the group of $A$-valued characters of $H$. 
\end{lemma}
The convolution inverse of such a $\phi$ is given by the series expansion
$\phi^{-1} = \sum_{n\geq 0} (e - \phi)^{*n}$, which is convergent for every $x$,
by induction on the grading, and because $e$ and $\phi$ agree on $B_0$.

\begin{prop}\label{PropI}
  If $\phi(x)=1$ for all group-like elements $x$, then the definitions of
  $\phi_-$ and $\phi_+$ from Proposition~\ref{RENH} make sense, and the
  conclusions there hold again.
\end{prop}

\begin{proof}
  The proof goes mostly as in the connected case, but using 
  \eqref{eq:split} instead of~\eqref{eq:primitive}.
  Again, since $e$ and $\phi$ agree on $B_0$, the basis of the recursion
  is clear, and we have $\phi_-(x)=e(x)=1$ for all group-like elements 
  $x$.  For the same reason, for $x$ homogeneous of degree $n>0$, 
  the $\Delta_{n,0}$ part of the
  comultiplication is zero.  For the middle part, note again that since $B_+ 
  \subset \Ker \epsilon$, we are left with $-R\big( \sum_{(x)} 
    \phi_-(x')\phi(x'')\big)$ just as in the connected case.  For
  the $\Delta_{0,n}$ part, observe again that $e$ and $\phi_-$ agree on the
  left-hand tensor factor.  Now $\phi_-\tensor \phi$ can be taken in two
  steps:
  $$
  B\tensor B \stackrel{\epsilon\tensor \id}\longrightarrow \ground \tensor B
  \stackrel{\eta_A\tensor \phi}\longrightarrow A \tensor A .
  $$
  By counitality of $\Delta$, the first step yields $1\tensor x$, and
  therefore the second step yields $1\tensor \phi(x)$, which finally
  multiplies to $\phi(x)$, hence altogether this part gives
  $-R\phi(x)$, just as in the connected case.  Therefore formula
  \eqref{phi-x} holds true again.
  The proof of the character property, $\phi_-(xy) = 
  \phi_-(x)\phi_-(y)$, follows the proof in the 
  connected case~\cite{Manchon:0408}, by induction on $\deg(x)+\deg(y)$.
  The only new ingredient needed is the following 
  lemma (trivial in the connected case), which is easily proved by 
  induction.
\end{proof}

\begin{lemma}
  With notation as above, if $x$ is group-like then for all $y$,
  $$
  \phi_-(xy) = \phi_-(yx) = \phi_-(y) .
  $$
\end{lemma}

\begin{blanko}{Example.}
  For $\deg(x)=1$, we have
  $$
  \phi_-(x) = -R(\phi(x)) \qquad \text{ and } \qquad
  \phi_+(x) = \phi(x) - R(\phi(x))
  $$
  just like in the connected case, even though $x$ cannot be assumed 
  to be primitive.
\end{blanko}

\section{Bialgebra renormalisation II}
\label{sec:bialgII}

The assumption that the (regularised) Feynman rules assign value $1$ to every
group-like element is perhaps not realistic from the viewpoint of physics
(cf.~Section~\ref{sec:ex} below for discussion).  We proceed to show that weaker
hypotheses are possible on the Feynman rules, provided some further conditions
are imposed on the bialgebra, to allow reduction to the previous case.

\begin{blanko}{Hypothesis II.}\label{condII}
  We assume that our filtered bialgebra $B$ is defined from combinatorial data in
  the following precise sense.  $B$ is the free vector space on a set $C$ of
  homogeneous `combinatorial elements', which is closed under multiplication, and also
  closed under comultiplication in the sense that for $x\in C$, all the terms in
  \eqref{eq:split} belong to $C\times C$.  
  In this situation, the key requirement we make is
  that for
  $x\in C$ of degree $n$, we have that both $\Delta_{0,n}(x)$ and $\Delta_{n,0}(x)$ are
  `indecomposable group-like', meaning that
  $$
  \Delta_{0,n}(x) = \ind(x) \tensor x \quad \text{ and } \quad 
  \Delta_{n,0}(x) = x\tensor \out(x)
  $$
  where both $\ind(x)$ and $\out(x)$ are group-like elements.
  It follows that the elements $x\in C_0=C\cap B_0$ are precisely the group-like
  elements. Hence Hypothesis II implies Hypothesis I.
\end{blanko}

The terminology `$\ind$' and `$\out$' is motivated mainly by
Example~\ref{ex:trees}: for $x$ a forest in the bialgebra of operadic trees,
$\ind(x)$ is the set of leaves and $\out(x)$ is the set of roots.
The terms $\ind(x) \tensor x + x\tensor \out(x)$ constitute the {\em
skew-primitive} part of the comultiplication,\footnote{The notion of
skew-primitive is well established in the Hopf algebra literature; see for example
\cite{Chin:brief}.} playing the role of the primitive part in the connected
case, $1\tensor x+x\tensor 1$.

\begin{lemma}\label{lem:deg(xy)}
  For $x,y \in C$, if $xy\in C_0$ then $x\in C_0$ and $y\in C_0$.
\end{lemma}
\begin{proof}
  If we had $\deg(x)>0$ then $\epsilon(x)=0$ and hence $\epsilon(xy)=0$,
  in contradiction with the fact that $xy$ is group-like.
\end{proof}
\begin{lemma}
  The assignments $\ind: C \to C_0$ and $\out: C \to C_0$ are idempotent
  monoid homomorphisms.
\end{lemma}

\begin{proof}
  We do the case of $\out$.
  We have $\Delta_{n,0}(xy) = xy \tensor \out(xy)$.  On the other hand,
  the $n+0$ part of $\Delta(x)\Delta(y) = \big( \cdots + x \tensor \out(x)\big)
  \big( \cdots + y \tensor \out(y)\big)$ contains the term $xy \tensor 
  \out(y)\out(y)$.
  It cannot contain other $n+0$ terms by Lemma~\ref{lem:deg(xy)}
\end{proof}

\begin{blanko}{Residue.}
  The monoid homomorphism $\out: C \to C_0$ extends to an
  algebra homomorphism
  $$
    \res: B  \to  B_0  .
  $$
\end{blanko}

\begin{blanko}{Calibration.}\label{RBsc}
  For $\phi \in \Lin(B, A)$, 
  define $\tilde\phi : C \to A$ by
  $$
  \tilde\phi := \frac{\phi}{\phi\circ\res}.
  $$
  For this to make sense, we need to assume that for all $x\in C$, 
  \begin{equation}\label{phires}
    \phi(\res(x)) \ \text{ divides } \ \phi(x).
  \end{equation}
  Extend linearly to $\tilde \phi : B
  \to A$.  Since $\res$ is a projection onto $B_0$, clearly $\tilde\phi$ sends
  group-like elements to $1$.  Clearly it is again multiplicative if $\phi$ is
  so.
\end{blanko}

Now the following is an immediate consequence of Proposition~\ref{PropI}.
\begin{prop}
  If $B$ satisfies Hypothesis II, and if $\phi: B \to A$ satisfies 
  \eqref{phires}, then with notation as above,
  the recursive definition
\begin{equation}\label{eqII}
  \phi_- \  := \ e + R(\phi_- * (e \!-\!\tilde\phi)).
\end{equation}
is meaningful, and  the renormalised Feynman rule
$$\phi_+ := \phi_- *\phi$$
maps $B_+$ to $A_+$.  If $\phi$ is character, then so are $\phi_-$ and 
$\phi_+$.
\end{prop}
It should be noted that the Bogoliubov counter term $\phi_-$ depends only on the calibrated
Feynman rule $\tilde \phi$, not on $\phi$ itself, but that the final renormalised
Feynman rule $\phi_+$ does take the full information in $\phi$ into account.


\section{Examples}

\label{sec:ex}

\begin{blanko}{Perturbative quantum field theory.}
  An inner line in a Feynman graph is called a {\em bridge} if removing it would 
  increase the number of connected components.  A graph is called 
  {\em 1PI} ($1$-particle irreducible) if it contains no bridges.
  We call a connected graph a {\em star} when it contains no inner lines 
  (hence is 1PI).
  
  In perturbative quantum field theory \cite{Connes-Kreimer:9912092},
  $H$ is a Hopf algebra of certain 1PI
  Feynman graphs, but {\em excluding stars}.  The
  comultiplication is given by
  $$
  \Delta(\Gamma) \ = \ 1 \tensor \Gamma 
  \ + \ \sum_{\gamma\subsetneqq \Gamma} \gamma \tensor \Gamma/\gamma
  \ + \ \Gamma \tensor 1 ,
  $$
  
  \vspace{-6pt}
  
  \noindent
  like for example
  \begin{center}
  \begin{texdraw}
   \htext (0 0) {$\Delta($}
   \rmove (4 0) \femkant 
     \rmove (27 0)
    \htext  {$) \ = $}
      \rmove (15 0)
      \htext{$1$}

      \rmove (10 0) \htext {$\tensor$}
      
      \rmove (5 0) \femkant
      
      \rmove (30 0) \htext {$+$}
      
      \rmove (10 0) \trekant 
      \rmove (18 0) \htext {$\tensor$}
      \rmove (5 0) \trekant

      \rmove (28 0) \htext {$+$}
      
      \rmove (12 0) \femkant       \rmove (24 0) \htext {$\tensor$}
      \rmove (10 0) \htext{$1$.}
 \end{texdraw}
  \end{center}
  Note that the comultiplication is defined `by hand' to have the form
  \eqref{eq:primitive}.  In a uniform description, the natural last term would
  be $\Gamma\tensor\Gamma/\Gamma = \Gamma\tensor\res \Gamma$, but the star $\res
  \Gamma$ was excluded!  One is led, as Manchon~\cite{Manchon:0408}, {\em not}
  to exclude the stars.  But then the small graphs $\gamma$ should be allowed to
  be stars too,\footnote{Manchon does actually not allow the $\gamma$ to be
  stars, but keeps the quotienting in the left-hand tensor factor.  Note also
  that the $\epsilon$ he indicates is not in fact the counit.} and in the end
  the comultiplication looks like this:
    \begin{center}
  \begin{texdraw}
   \htext (0 0) {$\Delta($}
   \rmove (4 0) \femkant 
     \rmove (27 0)
    \htext  {$) \ = $}
      \rmove (7 0)
      \rmove (7 0)\trevert 
      \rmove (7 0)\trevert 
      \rmove (7 0)\trevert 
      \rmove (7 0)\trevert 
      \rmove (7 0)\trevert

      \rmove (18 0) \htext {$\tensor$}
      
      \rmove (5 0) \femkant
      
      \rmove (29 0) \htext {$+$}
      
      \rmove (11 0) \trekant \rmove (9 0) \trevert \rmove (7 0) \trevert 
      \rmove (18 0) \htext {$\tensor$}
      \rmove (5 0) \trekant

      \rmove (28 0) \htext {$+$}
      
      \rmove (12 0) \femkant       \rmove (24 0) \htext {$\tensor$}
      \rmove (6 0) \trevert 
 \end{texdraw}
  \end{center}
  It is now a bialgebra $B$ rather than a Hopf algebra, since the stars (and
  disjoint unions of stars) are group-like.  One possible grading is by loop 
  number, with which it is clear that $B$
  satisfies Hypothesis II: for $x$ a graph, $\ind(x)$ is the set of vertices,
  and $\out(x)$ is the residue of $x$.
  The connected quotient is the usual Hopf algebra $H$.
  
  Mathematically this bialgebra has some pleasant features: for one thing, all
  the left-hand tensor factors have the same set of vertices as the original
  graph, and all the right-hand tensor factors have the same residue as the
  original graph.  Furthermore, the residue of each left-hand tensor factor
  matches precisely the set of vertices of the right-hand tensor factor.
  
  Physically, an attractive feature of $B$ is that it contains all the terms of
  the (bare) Lagrangian; these are the stars, including stars like
  \begin{texdraw}\setunitscale{1.2} \lvec (10 0)\move (4 1) \lvec (6, -1) \move
  (4 -1) \lvec (6, 1) \move (0 -2) \end{texdraw} for mass terms. 
  
  The Feynman rules (as defined in text books, independently of Hopf algebra
  viewpoints) naturally assign non-trivial amplitudes also to stars; these
  cannot be seen at the level of $H$.  What the Feynman rules exactly are in the
  Hopf algebra interpretation is more subtle than outlined in \ref{ex:QFT} (see
  Kreimer~\cite{Kreimer:0202110}, \S 1.3): in reality the Feynman rules depend
  on external momenta (and possibly other physical parameters), but it is a
  basic feature that this dependence is the same for a graph and for its
  residue, up to a scalar function (the so-called form factor).  Hence the
  divisibility assumption \eqref{phires} is validated, and these parameters
  cancel out in $\tilde\phi$, which then clearly sends group-like elements to
  $1$.  Kreimer writes in fact (\cite{Kreimer:0609004}, \S 3.3): ``Our Feynman
  rules [...]~are normalized to evaluate the tree-level term to unity.''  The
  calibration step \ref{RBsc} may be seen as a transparent formalisation of this
  assumption, building it into the abstract framework which makes sense also
  beyond the case of graphs.

  
  The bialgebra of Feynman graphs has the prospective of formulating more
  aspects of renormalisation inside it than are possible inside $H$.
  For instance, one may wish to write down a Dyson--Schwinger equation such as
  \begin{center}\begin{texdraw}
    \move (0 0) 
    \bsegment
    \lvec (12 0) \rlvec (8 8) \move (12 0) \rlvec (8 -8)
    \move (12 0) \fcir f:0.8  r:3  \lcir r:3
    \esegment
    
    \htext (32 0) {$=$}
    
       \move (40 0) \bsegment
       \move (2 0)
       \lvec (12 0) \rlvec (8 8) \move (12 0) \rlvec (8 -8) 
       \esegment
       
        \htext (70 0) {$+$}
   
	\move (82 0)
	\bsegment
       \move (0 0)
       \lvec (12 0) \rlvec (16 16) \move (12 0) \rlvec (16 -16) 
       \move (22 10) \lvec (22 -10)
       \move (12 0) \fcir f:0.8  r:3  \lcir r:3
       \move (22 10) \fcir f:0.8  r:3  \lcir r:3
       \move (22 -10) \fcir f:0.8  r:3  \lcir r:3
       \esegment
       
            \htext (139 0) {$+ \quad \cdots$}
  \end{texdraw}
  \end{center}
  which, as it stands, makes sense and has a solution (the Green function) in 
  $B$ (or rather in the completion of $B$), but not in $H$.
\end{blanko}

\begin{blanko}{Trees --- combinatorial versus operadic.}\label{ex:trees}
  In the usual Connes--Kreimer Hopf algebra of rooted trees
  \cite{Connes-Kreimer:9808042}, 
  also called the Butcher--Connes--Kreimer Hopf algebra,
  the
  trees are combinatorial trees (such as \
  \raisebox{1pt}{\begin{texdraw}\smalldot\end{texdraw}} , \raisebox{-3pt}{
    \begin{texdraw}\twonodetree\end{texdraw}} , \raisebox{-3pt}{
    \begin{texdraw}\vtree\end{texdraw}}  ), and the admissible cuts used to define the
  comultiplication actually delete edges rather than cutting them:
  $$
  \Delta( \raisebox{-3pt}{
    \begin{texdraw}\twonodetree\end{texdraw}}
   \ ) 
\   =\  1 \tensor  \raisebox{-3pt}{
    \begin{texdraw}\twonodetree\end{texdraw}}
   \  \ + \ \
  \raisebox{1pt}{\begin{texdraw}\smalldot\end{texdraw}} \ \tensor \
  \raisebox{1pt}{\begin{texdraw}\smalldot\end{texdraw}} \ 
   \ + \ \raisebox{-3pt}{
    \begin{texdraw}\twonodetree\end{texdraw}}
   \  \tensor 1 .
  $$
  
  For an operadic interpretation of Hopf algebra renormalisation (as hinted at
  in many papers by Kreimer and his collaborators,
  e.g.~\cite{Bergbauer-Kreimer:0506190}), the natural trees to consider are {\em
  operadic trees} (i.e.~with open-ended edges (leaves)).  These form a bialgebra
  rather than a Hopf algebra, cf.~\cite{Kock:1109.5785}: the comultiplication is
  exemplified by
  
  \vspace{-4pt}
  
  \begin{center}\begin{texdraw}
    \htext (0 0) {$\Delta($}
   \rmove (12 0) \twotree
     \rmove (15 0)
    \htext  {$) \ \ = $}
        \rmove (23 0) \zerotree
    \rmove (3.5 0) \zerotree
    \rmove (3.5 0) \zerotree

          \rmove (11 0) \htext {$\tensor$}
   \rmove (14 0) \twotree

	            \rmove (18 0) \htext {$+$}

    \rmove (20 0) \onetree  \rmove (5 0) \zerotree 
    
    \rmove (12 0) \htext {$\tensor$}
   \rmove (12 0) \onetree
   
      \rmove (19 0) \htext {$+$}

    \rmove (21 0) \twotree
    
    \rmove (13 0) \htext {$\tensor$}
   \rmove (10 0) \zerotree
  \end{texdraw}\end{center}
  The nodeless trees and forests are the group-like elements.
  The bialgebra of operadic trees~\cite{Kock:1109.5785}
  is easily seen
  to satisfy Hypothesis II: for $x$ a forest, $\ind(x)$ is the forest
  consisting of its leaves, and $\out(x)$ is the forest consisting of its roots.
  There is a bialgebra homomorphism from operadic trees 
  to combinatorial trees
  given by taking core (i.e.~shaving off leaves and root
  \cite{Kock:1109.5785}); this is a more drastic quotient than
  just collapsing group-like elements.
  
  In analogy with the case of graphs, note that in the comultiplication formula
  in $B$, every left-hand tensor factor has the same leaf profile as the
  original tree or forest, while each of the right-hand tensor factors has the same root
  profile as the original tree or forest; again in each term, the root profile of the
  left-hand tensor factor matches the leaf profile of the right-hand tensor
  factor.
  
  We mention in passing that these strict typing constraints in the
  comultiplication formula were found important in recent
  work~\cite{GalvezCarrillo-Kock-Tonks:1207.6404} establishing a Fa\`a di Bruno
  formula for the Green function in the bialgebra of operadic trees, in analogy
  with similar formulae found by van Suijlekom~\cite{vanSuijlekom:0807} in the
  case of graphs.  The Green function is the sum of all operadic trees, weighted
  by symmetry factors (and it is crucial for these symmetry factors to come out
  right to use operadic trees rather than combinatorial trees); it appears as
  solution to a certain abstract combinatorial 
  Dyson--Schwinger equation in the category of groupoids.
  Further relationships with Category Theory and 
  Logic are explored in \cite{Kock:1109.5785} and \cite{Kock:DSE-W}.
\end{blanko}

\section{Coalgebra renormalisation and M\"obius inversion}
\label{sec:M}

In the above account of bialgebra renormalisation, the multiplication in $B$ did
not play the most important role: all the results can be stated for coalgebras
instead of bialgebras --- except of course that the notions of multiplicativity
and characters do not make sense any more.  We have:

\begin{prop}
  Let $C$ be a coalgebra satisfying Hypothesis II, and assume $\phi: C \to A$
  satisfies \eqref{phires}.  Then the recursive definition
  \begin{equation}\label{phiC}
    \phi_- \  := \ e + R(\phi_- * (e\!-\!\tilde\phi))
  \end{equation}
  makes sense, and 
  $$
  \phi_+ := \phi_- * \phi
  $$ 
  maps $C_+$ to $A_+$.
\end{prop}
Note that since there is no multiplicativity to establish, it is not
essential that $R$ be Rota--Baxter, it suffices to be idempotent.

\begin{blanko}{Incidence coalgebras of M\"obius categories.}
  \label{inc}
  Two classical settings for incidence (co)algebras and M\"obius inversion are
  locally finite posets (Rota et al.), and monoids with the finite decomposition
  property (Cartier--Foata).  An elegant common
  generalisation is Leroux's notion of M\"obius category (for which we refer to
  \cite{Lawvere-Menni} for a modern
  treatment).  Very briefly, a M\"obius category is a category $\CC$
  subject to some finiteness conditions to make the following constructions
  make sense. The incidence coalgebra of $\CC$ is the vector
  space $C$ spanned by the arrows of $\CC$ (the arrows are the combinatorial 
  elements in the sense of \ref{condII}), with coalgebra structure
  given by the formula
  $$
  \Delta(f) = \sum_{b\circ a = f} a \tensor b .
  $$
  The sum is over all pairs of composable arrows whose composite is the arrow 
  $f$.
  
  (The classical settings are special cases of M\"obius categories, by
  interpreting a poset as a category in which there is one arrow $x\to y$
  whenever $x \leq y$, and by interpreting a monoid as a category with only one
  object (the monoid elements being then the arrows).)

  A coalgebra filtration of $C$ is given by the maximum length of effective
  chains of arrows (i.e.~not involving identity arrows) that compose to a given
  arrow.  (The M\"obius condition is equivalent to the existence of this
  filtration.)  Clearly $C_0$ is spanned by the identity arrows, and these are
  group-like since the only factorisation is $\id = \id\circ \id$.  Finally, for
  any arrow $f:x\to y$ the trivial factorisations $f = f \circ \id_x$ and $f =
  \id_y \circ f$ constitute the only $0+n$ and $n+0$ splittings, so as to verify
  Hypothesis II.  In conclusion:
\end{blanko}

\begin{prop}
  The incidence coalgebra of a M\"obius category always
  satisfies Hypothesis II.
\end{prop}

\begin{blanko}{Remark.}\label{3=1+1}
  The following configuration of arrows in a M\"obius category 
  illustrates the need for filtering rather than bona fide grading:
  $$
  
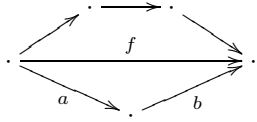
\xymatrix @! @R=10pt @C=5pt {
  && \cdot \ar[rr] && \cdot \ar[rrd] && \\
  \cdot \ar[rru] \ar[rrrrrr]^f \ar[rrrd]_a &&&&&& \cdot \\
  &&& \cdot \ar[rrru]_b &&&
  }
  $$
  Clearly $\deg(f) = 3$, but $\Delta(f)$ contains the term $a\tensor b$
  of degree splitting $1+1$.
\end{blanko}

\begin{blanko}{M\"obius inversion.}
  A very special case of renormalisation in the coalgebra setting is M\"obius
  inversion.  Let $C$ be the incidence coalgebra of a M\"obius category.  Let
  $A=\ground$ be the trivial Rota--Baxter algebra (the ground field with $R$ the
  identity map).  Take $\phi$ to be the zeta function, $\zeta(x) = 1$ for all
  combinatorial elements $x$.  Then $\phi_-$ is the M\"obius function $\mu$
  (i.e.~the inverse to $\zeta$ in the convolution algebra $\Lin(C,\ground)$).
  Indeed, the standard formula for M\"obius inversion (see \cite{Stanley:volI}
  for the poset case)
  $$
  \mu(\id) = 1, \qquad \mu(x) = - \sum_{\substack{ab=x\\ b\neq \id}} \mu(a)
  $$
  can be written
  $$
  \mu = \epsilon + \mu * ( \epsilon- \zeta)
  $$
  which is precisely \eqref{phiC}.  (From a renormalisation viewpoint, this is 
  a very degenerate case: $A_+ = \{0\}$, and the `renormalised 
  zeta function' is just $\phi_+=\epsilon$.)
  
  In the incidence coalgebra situation, dividing out by the co-ideal spanned by
  $1-x$ for $x$ group-like corresponds to considering certain reduced incidence
  coalgebras, but unlike in the bialgebra case this does not imply the existence
  of an antipode. 
  When the antipode exists, of course M\"obius inversion is 
  nothing more than applying the antipode, but the M\"obius inversion formula
  is more general. 
\end{blanko}

\begin{blanko}{M\"obius decomposition spaces.}
  A far-reaching generalisation of the notion of M\"obius category was
  introduced recently in \cite{Galvez-Kock-Tonks:1404.3202} under the name
  M\"obius decomposition space.  We shall not reproduce the definition here, but only
  mention three facts: (1) The above notions and results generalise readily to
  M\"obius decomposition spaces;  (2) a monoidal structure on a M\"obius
  decomposition space makes the resulting incidence coalgebra a bialgebra; and
  (3) in fact the bialgebras of graphs and trees are examples of incidence
  bialgebras of monoidal M\"obius decomposition spaces.
\end{blanko}

\bigskip

\noindent \textbf{Acknowledgements.} I am grateful to Kurusch Ebrahimi-Fard for
his patience, over the past five years, in explaining to me the rudiments of
renormalisation theory, together with many related fascinating topics, and more
specifically for his help with the present note.  I also thank Marc Bellon for
useful feedback.  This research has been sponsored by grant number
MTM2013-42293-P of Spain.


\noindent
Joachim Kock \texttt{<kock@mat.uab.cat>}\\
Departament de matem\`atiques\\
Universitat Aut\`onoma de Barcelona\\
08193 Bellaterra (Barcelona)\\
SPAIN.

\end{document}